\documentclass{fundam}
\usepackage{url} 
\usepackage{graphicx}
\usepackage{mathrsfs}
\usepackage{amssymb}
\usepackage{amsmath}
\usepackage{algpseudocode}

\RequirePackage{xspace}
\let\lg\relax
\DeclareMathOperator{\lg}{ld}

\DeclareMathOperator{\bin}{bin}
\DeclareMathOperator{\cod}{cod}
\newcommand{\dfa}{\textsc{dfa}\xspace}
\newcommand{\nfa}{\textsc{nfa}\xspace}

\newcommand\subtext[1]{\mathrm{#1}}
\newcommand\set[1]{\{#1\}}
\newcommand\sett[2]{\{\,#1\ \mid\ #2\,\}}
\newcommand{\trans}{\textsc{iufst}}
\newcommand{\htrans}{\textsc{-iufst}}
\newcommand{\ntrans}{\textsc{niufst}}
\newcommand{\hntrans}{\textsc{-niufst}}

\newcommand{\lba}{\textsc{lba}\xspace}

\newcommand{\cupdot}{\mathbin{\mathaccent\cdot\cup}}
\newcommand{\blank}{\raisebox{\depth}{\texttt{\char 32}}}

\newcommand{\dollar}{\texttt{\$}}

\newcommand{\border}{\texttt{\#}}
\newcommand{\rightend}{\mathord{\vartriangleleft}}
\newcommand{\leftend}{\mathord{\vartriangleright}}
\begin{document}

\setcounter{page}{337}
\publyear{22}
\papernumber{2113}
\volume{185}
\issue{4}

 \finalVersionForARXIV

\title{Computational and Descriptional Power of Nondeterministic \\
                      Iterated Uniform Finite-State Transducers\footnote{A preliminary version of this work was
                 presented at the \emph{16th Conference on Computability in Europe (CiE 2020), Fisciano,
                 Italy, June 29th -- July 2nd, 2020, and it is published in
                 {\rm \cite{kutrib:2020:dniufstcdp:proc}}.}}  }

\author{Martin Kutrib\thanks{Address of correspondence:  Institut f\"ur Informatik, Universit\"at Giessen,
                                Arndtstr.~2, 35392 Giessen, Germany.\newline \newline
          \vspace*{-6mm}{\scriptsize{Received June 2022; \ accepted May  2022.}}}, Andreas Malcher
          \\
  Institut f\"ur Informatik, Universit\"at Giessen\\
  Arndtstr.~2, 35392 Giessen, Germany\\
   \{kutrib, andreas.malcher\}@informatik.uni-giessen.de
\and Carlo Mereghetti, Beatrice Palano\\
  Dip. di Informatica ``G.\ degli Antoni'',  Universit\`a degli Studi di Milano\\
  via Celoria~18, 20133 Milano, Italy\\
 \{carlo.mereghetti, beatrice.palano\}@unimi.it
 }

\runninghead{M.~Kutrib et al.}{Iterated Uniform Finite-State Transducers}

\maketitle

\vspace*{-8mm}
\begin{abstract}
An iterated uniform finite-state transducer ($\trans$) runs the same length-preserving
transduction, starting with a sweep on the input string and then iteratively sweeping
on the output of the previous sweep.
The $\trans$ accepts the input string by halting in an accepting state at the
end of a sweep.
We consider both the deterministic ($\trans$) and nondeterministic ($\ntrans$) version
of this~device. We show that constant sweep bounded $\trans$s and $\ntrans$s accept all
and only regular languages. We study the state complexity of removing nondeterminism as well as
sweeps on constant sweep bounded $\ntrans$s, the descriptional power of constant sweep
bounded $\trans$s and $\ntrans$s with respect to classical models of finite-state automata,
and the computational complexity of several decidability questions.
Then, we focus on non-constant sweep bounded devices, proving the existence of a proper
infinite nonregular language hierarchy depending on the sweep complexity both in the
deterministic and nondeterministic case.
Though $\ntrans$s are ``one-way'' devices we show that they characterize
the class of context-sensitive languages, that is, the complexity class $\mathsf{DSpace}(\text{lin})$.
Finally, we show that the nondeterministic devices are more powerful than their
deterministic variant for a sublinear number of sweeps that is at least
logarithmic.

\medskip\noindent
\textbf{Keywords:}
Iterated transducers, nondeterminism, descriptional complexity, sweep complexity, language hierarchies.
\end{abstract} 

\section{Introduction}\label{sec:intro}\vspace*{-1mm}

The notion of an iterated uniform finite-state transducer ($\trans$) has been
introduced in~\cite{KMMP19} (see~\cite{kutrib:2022:dciufst} for the journal version) and
can be described as a finite transducer that iteratively
sweeps from left to right over the input tape while performing the same length-preserving
transduction.
In particular, the output of the previous sweep is taken as input for every new sweep.
This model is motivated by typical applications of transducers or cascades of transducers,
where the output of one transducer is used as the input for the next transducer.
For example, finite state transducers are used for the lexical analysis of computer programs
and the produced output is subsequently processed by pushdown automata for the syntactical
analysis.
In~\cite{Friburger:fstcteneit:2004}, cascades of finite-state transducers are used in
natural language processing.
Another example is the Krohn-Rhodes decomposition theorem which shows that every regular
language is representable as the cascade of several finite state transducers, each one
having a ``simple'' algebraic structure~\cite{Ginzburg:atoa:1968,Hartmanis:1966:astosm:book}.
Finally, it is shown in~\cite{Citrini:odmpa:1986} that cascades of deterministic pushdown
transducers lead to a proper infinite hierarchy in between the deterministic context-free
and the deterministic context-sensitive languages with respect to the number of transducers
involved.

In contrast to all these examples and other works in the literature (see, e.g.,
\cite{Bordihn:istalgd,manca:2001:gpit,pierce:2011:dpoilpt}), where the subsequently applied
transducers are in principle different and not necessarily length-preserving,
the model of $\trans$s introduced in~\cite{KMMP19} requires that the same transducer
is applied in every sweep and that the transduction is deterministic and length-preserving.
More precisely, an $\trans$ works in several sweeps on a tape which initially contains the
input string concatenated with a right endmarker.
In every sweep, the finite state transducer starts in its initial state at the first
tape cell, is applied to the tape, and prints its output on the tape.
The input is accepted, if the transducer halts in an accepting state at the end of a sweep.
In~\cite{KMMP19}, $\trans$s both with a constant number and a non-constant number of sweeps
are investigated. In the former case, it is possible to characterize exactly the set
of regular languages and upper and lower state bounds for converting $\trans$s into
deterministic finite automata ($\dfa$s) and vice versa are established.
Furthermore, as always done for several models (see, e.g.,
\cite{bednarova:2012:scbochdpda,BMP10,BMP17,BGMP13,DBLP:journals/jcss/BednarovaGMP17,JMMP13}),
the state complexity of language operations,
that is, the costs in terms of the number of states needed for union, intersection,
complementation, and reversal, is investigated in depth. Finally, the usually studied
decidability questions such as emptiness, finiteness, equivalence, and inclusion are
proved to be {\sf NL}-complete, showing that these questions have the same computational
complexity as for $\dfa$s.
For the case of a non-constant number of sweeps, the situation is quite different.
It is shown that a logarithmic number of sweeps is sufficient to accept unary non-semilinear
languages, while with a sublogarithmic number of sweeps only regular languages can be
accepted.
Moreover, the existence of a finite hierarchy with respect to the number of sweeps is
obtained. Finally, all usually studied decidability questions are shown to be undecidable
and not even semidecidable for $\trans$s performing at least a logarithmic number~of~sweeps.

In this paper, we enhance the model of $\trans$s by {\em nondeterminism},
thus obtaining their \emph{nondeterministic} version ($\ntrans$s).
As in \cite{KMMP19}, we are interested in
$\ntrans$s exhibiting both a constant and non-constant number of sweeps.

Constant sweep bounded $\ntrans$s are proved to accept exactly regular languages.
So, their ability of representing regular languages in a very succinct way turns out to
be worth investigating, as well as comparing such an ability with that of other more
traditional models of finite-state automata.
This type of investigation, whose importance is witnessed by a well consolidated trend in
the literature, focuses on the number of states for representing languages and
belongs to the area of descriptional complexity.
Being able to have ``small'' devices representing/accepting certain languages, leads to
relevant consequences either from a practical point of view (less hardware needed to
construct such devices, less energy absorption, less cooling problems, {\em etc.}), and
from a theoretical point of view (higher manageability of proofs and representations for
languages, reductions of difficult problems on general computing devices to the same
problems on simpler machines, {\em etc.}).
The reader is referred to, e.g.,~\cite{holzer:2010:dcais}, for a thoughtful survey on
descriptional complexity and its consequences.

Non-constant sweep bounded $\ntrans$s are then studied for their computational power,
i.e., the ability of accepting language families. In particular, such an ability is
related to the number of sweeps as a function of the input~length.

After defining $\ntrans$s in Section~\ref{sect:def},
we discuss in detail an example that demonstrates the state number advantages of $\ntrans$s with a
constant number of sweeps in comparison with its deterministic variant and the
classical models of deterministic and nondeterministic finite automata ($\nfa$s).
Precisely, we exhibit a language accepted by an $\ntrans$ such that any equivalent $\trans$
requires exponentially more states and sweeps, while any equivalent \nfa\ (resp., \dfa)
requires exponentially (resp., double-exponentially) more states.

In Section~\ref{sect:constant}, we consider $\ntrans$s with a constant
number of sweeps in more generality. By evaluating the state cost of sweep removal, we
show that any $\ntrans$ featuring~$n$ states and~$k$ sweeps can be simulated by
a~$2n^k$-state \nfa, and hence by a $2^{2n^k}$-state \dfa\ as well.
Next, we exhibit a unary (resp., binary) language to establish lower bounds for
the obtained size blow-up
for turning a constant sweep $\ntrans$ into an equivalent \nfa\ (resp., \dfa).
Moreover, we study the computational complexity of several decidability
questions for $\ntrans$ with $k$ sweeps and obtain {\sf NL}-completeness results
for the questions of emptiness, finiteness, and infiniteness, whereas the questions
of universality, inclusion, and equivalence are shown to be {\sf PSPACE}-complete.

In the last two sections, we consider $\ntrans$s with a non-constant
number of sweeps. First, we establish in Section~\ref{sec:infinite-hierarchy} an infinite
proper hierarchy with respect to the number of sweeps. Interestingly, this result also
extends the known finite hierarchy in the deterministic case to an infinite
hierarchy. Then we show that $\ntrans$s can simulate linear bounded automata,
though $\ntrans$s are ``one-way'' devices where the information flow is from
left to right only. So, $\ntrans$s whose sweep complexity is not bounded a
priori characterize the class of context-sensitive languages, that is, they
capture the complexity class $\mathsf{DSpace}(\text{lin})$.

Finally, we study in Section~\ref{sec:ndet-beats-det} the question of whether the
nondeterministic model is more powerful than the deterministic model. We get that the
question can be answered in the affirmative if at least a logarithmic number of sweeps is
provided. Moreover, we show that nondeterminism cannot be matched in power by the
deterministic paradigm even if a sublinear number of sweeps is~given.

\section{Definitions and preliminaries}\label{sect:def}

We denote the set of positive integers and zero by ${\mathbb{N}}$.
Set inclusion is denoted by~$\subseteq$ and strict set inclusion by $\subset$.
Given a
set~$S$, we write~$2^{S}$ for its power set and~$|S|$ for its cardinality.
Let~$\Sigma^*$ denote the set of all words over the finite
alphabet~$\Sigma$.  The \emph{empty word} is denoted by~$\lambda$, and
$\Sigma^+ = \Sigma^* \setminus \{\lambda\}$.  The length of a word $w$ is denoted
by $|w|$. By $\lg n$ we denote the logarithm of~$n$ to base~$2$.

Roughly speaking, an iterated uniform finite-state transducer is a
finite-state transducer which processes the input
in multiple passes (also sweeps). In the first pass it reads the
input word followed by an endmarker and emits an output
word. In the following passes it reads the output word of
the previous pass and emits a new output word.
It can be seen as a restricted variant of a one-tape Turing machine.
The number of passes taken, the \emph{sweep complexity},
is given as a function of the length of the input.
Here, we are interested in weak processing devices: we will
consider length-preserving finite-state transducers,
also known as Mealy machines \cite{Mealy:1955:amfssc}, to be iterated.

\medskip
Formally, we define a \emph{nondeterministic iterated uniform finite-state transducer}
($\ntrans$) as a
system $T=\langle Q,\Sigma,\Delta,q_0,\rightend,\delta,F\rangle$, where:
\begin{itemize}
\item $Q$ is the finite set of \emph{internal states},
\item $\Sigma$ is the set of \emph{input symbols},
\item $\Delta$ is the set of \emph{output symbols},
\item $q_0\in Q$ is the initial state,
\item $\rightend\in\Delta\setminus \Sigma$ is the \emph{endmarker},
\item $F\subseteq Q$ is the set of \emph{accepting states},
\item $\delta\colon Q\times(\Sigma\cup\Delta)\to 2^{Q\times \Delta}$
is the partial \emph{transition function}.
\end{itemize}

The $\ntrans$ $T$ \emph{halts} whenever the transition function is undefined
or whenever it enters an accepting state at the end of a sweep.
Since the transducer is applied in multiple passes, that is,
in any but the initial pass it operates on an output of the previous pass,
the transition function depends on input symbols from
$\Sigma\cup\Delta$.
We denote by~$T(w)$ the set of possible outputs produced by $T$ in a complete
sweep on input $w \in (\Sigma\cup\Delta)^*$.
During a computation on input $w\in\Sigma^*$, the $\ntrans$ $T$
produces a sequence of words $w_1,\ldots,
w_i,w_{i+1},\ldots\in(\Sigma\cup\Delta)^*$, where
$w_1\in T(w\rightend)$ and \mbox{$w_{i+1}\in T(w_i)$} for~$i\geq 1$.

An $\ntrans$ is said to be \emph{deterministic} ($\trans$)
if and only if $|\delta(p,x)|\leq 1$, for all $p\in Q$ and
$x\in\Sigma\cup\Delta$. In this case, we simply write
\mbox{$\delta(p,x)=(q,y)$} instead of
$\delta(p,x)=\{(q,y)\}$ assuming that the transition function is a
mapping \mbox{$\delta\colon Q\times(\Sigma\cup\Delta)\to {Q\times \Delta}$.}

Now we turn to language acceptance. With respect to nondeterministic
computations and some complexity bound,
in the literature several acceptance modes are considered. For example,
a machine accepts a language in the \emph{weak mode}, if for any input $w\in
L$ there is an accepting computation that obeys the complexity bound.
Language $L$ is accepted in the \emph{strong mode}, if the machine
obeys the complexity bound for all computations (accepting or not) on all
inputs.
Here we deal with the number of sweeps as (computational) complexity measure.
The weak mode seems too optimistic for this measure, while the
strong mode seems too restrictive. Therefore, here we consider an intermediate mode,
the so-called accept mode. A language is accepted in the \emph{accept mode}
if all accepting computations obey the complexity bound
(see~\cite{mereghetti:2008:tdcpstm} for separation of these modes with respect
to space complexity).

A computation is halting if there exists an $r\geq 1$ such that
$T$ halts on $w_{r}$, thus performing $r$ sweeps.
The input word $w\in\Sigma^*$ is \emph{accepted} by $T$ if at least one
computation on $w$ halts at the end of a sweep in an accepting state.
That is, the initial input is a word over the input alphabet $\Sigma$ followed by
the endmarker, and there is an output computed after $r-1$ sweeps that
drives $T$ in a complete final sweep where it halts in an accepting state.
Otherwise, it is \emph{rejected}.
Note that the output of the last sweep is not used.
The language accepted by $T$ is the set $L(T)\subseteq\Sigma^*$ defined as
\mbox{$L(T)=\{\, w \in \Sigma^* \mid w \text{ is accepted by }T\,\}$.}

Given a function $s\colon\mathbb{N}\to \mathbb{N}$,
an iterated uniform finite-state transducer $T$ is said to be
of \emph{sweep complexity}~$s(n)$ if for all $w\in L(T)$
all accepting computations on $w$ halt after at most~$s(|w|)$ sweeps.
In this case, we add the prefix $s(n)\textsc{-}$ to the notation of the device.
It is easy to see that $1\htrans$s (resp., $1\hntrans$s) are essentially
deterministic (resp., nondeterministic) finite-state automata
($\dfa$s and $\nfa$s, respectively).

Throughout the paper, two accepting devices are said to be
\emph{equivalent} if and only if they accept the same language.

We chose to denote our transducers as ``uniform'' since they perform the same
transduction at each sweep: they always start from the same initial state on the leftmost
tape symbol, operating the same transduction rules at every computation step. Yet, we
quickly observe that an $\ntrans$  is clearly a restricted version of a linear bounded
automaton (see, e.g., \cite{Hopcroft:1979:itatlc:book}). So, any language accepted by an
$\ntrans$ is context sensitive and it will turn out that the converse is also
true.

\subsection*{Accepting languages by iterated transductions: an example}

In order to clarify the notion of acceptance by iterated transduction,
for any integer $k\geq 2$, we consider the block language
$$
B_k=\{\,u_1\border u_2\border\cdots\border u_m \mid u_i\in\set{0,1}^k,\, m> 1,\, \exists i< m\colon u_i=u_m\,\}.
$$
By counting arguments, it is not hard to see that at least $2^{2^k+1}$ states are necessary
to accept $B_k$ by a $\dfa$.
On the other hand, an exponentially smaller $\nfa$ $A$ may accept~$B_k$
as follows.
\begin{enumerate}
\item In a first phase, on each block in the input string, $A$ stores the block in its
finite control and then nondeterministically decides whether to keep the block or to ignore
it. Along this phase, the correct block structure of the input
scanned so far is checked by~$A$ as well. This phase takes $2^{k+1}$ states.

\item Once $A$ decides to keep a block, say $u$, in its finite control, a second phase
starts in which $A$ scans the rest of the input checking the correct block structure and
guessing another block $w$ to be matched symbol-by-symbol against $u$. If the matching is
successful and $w$ is the last block, then $A$ accepts.
This phase takes at most $2^{k+1}\cdot (k+1)$ states.
\end{enumerate}

Altogether, the \nfa\ $A$ features $2^{k+1}+2^{k+1}\cdot (k+1)=2^{k+1}\cdot (k+2)$ states.

\medskip
Indeed, $A$~can also be seen as a $2^{k+1}\cdot (k+2)$-state $1\hntrans$ which outputs the
scanned symbol at each step.
However, paying by the number of sweeps (see, e.g., \cite{malcher:2012:dctwpdarhr}), we can
build a $k\hntrans$ $T$ for $B_k$ with only $O(k)$ states.
Informally:
\begin{enumerate}
\itemsep=0.9pt
\item In a first sweep, $T$ checks the correct block structure of the input string,
nondeterministically chooses two blocks to be matched symbol-by-symbol, and compares the
first symbol of the two blocks by storing the first symbol of the first block in its finite
control and replacing these two symbols with a blank symbol.

\item At the $i$th sweep, $T$ checks the $i$th symbol of the two blocks chosen in the first
sweep by storing and blank-replacing symbols as explained at the previous point.
To distinguish the first sweep (where both nondeterministic block choices and symbol
comparisons take place) from the others (where only symbol comparisons take place),
a special symbol  can replace the first input symbol at the beginning of the first sweep.
\end{enumerate}

It is not hard to see that $O(k)$ states are needed to check the input formatting along the
first sweep, and that a constant number of states suffices to blank-replacing and comparing
input symbols. Indeed, after $k$ sweeps all nondeterministically chosen block symbols are
compared so that $T$ may correctly accept or reject.
This gives the claimed state and sweep bounds for $T$.

We remark that: (i) a $2^k(k+4)$-state $2^k\htrans$ is designed in~\cite{kutrib:2022:dciufst} for $B_k$,
(ii)~$2^{2^k+1}$ states are necessary to accept $B_k$ by a $\dfa$, and that
(iii)~$2n^k$ states are sufficient for a \dfa\ to simulate an $n$-state $k$\htrans\
\cite{kutrib:2022:dciufst}.
These facts, together with the above designed $O(k)$-states $k$\hntrans, show that $\ntrans$s
can be exponentially more succinct than $\trans$s either in the number of states, or in the
number of sweeps, or possibly both.
Indeed, we also have that $\ntrans$s can be exponentially more succinct than $\nfa$s and
double-exponentially more succinct than $\dfa$s.
\smallskip

In the next section, we approach more generally the analysis of the descriptional power of
$\ntrans$s with respect to their deterministic counterparts and classical finite-state
models.

\section{$\ntrans$s with a constant number of sweeps}\label{sect:constant}
\subsection{Reducing sweeps and removing nondeterminism}

Let us begin by showing how to reduce sweeps from $\ntrans$s and evaluate the state cost of
reduction. We will then use this construction to reduce constant sweep bounded $\ntrans$s
to one sweep $\ntrans$s, thus obtaining equivalent $\nfa$s whose number of states will be
suitably bounded.

\begin{theorem}\label{the:sw}
Let $n,k,i>0$ be integers and $i \le k$. Every $n$-state $k\hntrans$ (resp., $k\htrans$) can be converted
to an equivalent $2n^i$-state $\lceil\frac{k}{i}\rceil\hntrans$
(resp., $\lceil\frac{k}{i}\rceil\htrans$).
\end{theorem}

\begin{proof}
The principal idea of the construction of an equivalent
$\lceil\frac{k}{i}\rceil\hntrans$ $T'$
is that the state set of~$T'$
is used to simulate~$i$ sweeps of the given $k\hntrans$~$T$ in parallel, step by step.
There may occur the problem that~$T$ gets stuck in some sweep $j\leq k$,
but an earlier sweep $1\leq \ell<j$ ends accepting. To cope with this problem, we
enforce~$T'$ on the one hand to continue the simulation of all sweeps $1\leq \ell<j$
and, on the other hand, to remember that the results of all sweeps after~$j$ are
no longer relevant for the simulation. Hence, these sweeps are simulated by entering
some dummy state~$d$. In addition, a symbol~$d$ is printed on the tape that ensures
that all later sweeps in~$T'$ are eventually performed in the dummy state~$d$ as well.
The formal construction is as follows.

Let $T=\langle Q,\Sigma,\Delta,q_0,\rightend,\delta,F\rangle$
be a $k\hntrans$ with $|Q|=n$.
The state set $Q'$ of an equivalent $\lceil\frac{k}{i}\rceil\hntrans$~$T'$ is defined as
$Q'=\bigcup_{t=0}^i Q^t\times \{d\}^{i-t}$, where~$d$ is a
new state not belonging to~$Q$.
The initial state $q'_0$ of $T'$ is $(q_0,q_0,\ldots,q_0)$.
The output alphabet is $\Delta'=\Delta \cup \{d\}$, where~$d$ is a
new output symbol not belonging to~$\Delta$.
The transition function $\delta'\colon Q'\times(\Sigma \cup \Delta') \to 2^{Q'\times \Delta'}$ of $T'$ is defined
by a procedure that determines the successor states and the output symbols.
For $x\in\Sigma\cup \Delta'$ we obtain
$$
((r_1,r_2,\dots,r_i),y_i)\in \delta'((s_1,s_2,\dots,s_i),x)
$$
as follows.\smallskip

\begin{algorithmic}[1]
\State $y_0 := x$; \quad dummynow := false;
\For{$t=1$ \textbf{to} $i$}
\If{dummynow \textbf{or} $s_t=d$ \textbf{or} $y_{t-1}=d$}
   \State $(r_t,y_t) := (d,d)$;
\Else
   \State $S := \delta(s_t,y_{t-1})$;\label{alg:line-s}
   \If{$S\neq \emptyset$}
      \State \textbf{guess} $(r_t,y_t)\in S$;
   \Else
       \State $(r_t,y_t) := (d,d)$;
       \State dummynow := true;
   \EndIf
\EndIf
\EndFor
\end{algorithmic}\smallskip

If $T$ accepts (for the first time) at the end of a sweep~$j$, the $\ntrans$ $T'$
can simulate at least these~$j$ sweeps of~$T$ successfully.
So, any state
$(r_1,r_2,\dots,r_i)\in Q'$ with an $r_t\in F$ is accepting for $T'$.
This works well if $i$ divides $k$, since in this case any of the~$i$ sweeps
simulated in~$T'$ corresponds to an original sweep in~$T$.
If $i$ does not divide $k$, then in the last, the $\lceil \frac{k}{i} \rceil$th
sweep of~$T'$ only the first $k-\lfloor \frac{k}{i} \rfloor\cdot i$
simulated sweeps of~$T$ have to be considered, since the remaining sweeps simulated
do not exist in~$T$. However, since all words accepted by~$T$ are accepted after
at most~$k$ sweeps, these additional sweeps can never lead to an erroneous acceptance.
For the number of states in~$T'$ we have, for $n\geq 2$,
$|Q'| = \sum_{t=0}^i n^t =
\frac{n^{i+1}-1}{n-1} \leq \frac{n^{i+1}}{n/2}= 2n^i$.

\medskip
If $T$ is deterministic,
the set $S$ in line~\ref{alg:line-s} is either empty or a singleton.
So, the guess in the former case boils down to select the sole element
deterministically. That is, the construction preserves determinism.
\end{proof}

The sweep reduction presented in Theorem~\ref{the:sw} can directly be used to transform
constant sweep bounded $\ntrans$s into equivalent
$\nfa$s.

\begin{theorem}\label{the:trans}
Let $n,k>0$ be integers. Every $n$-state $k\hntrans$ can be converted to an
equivalent $\nfa$ with at most $2n^k$ states.
\end{theorem}

\begin{proof}
Given an $n$-state $k\hntrans$~$T$ over the input alphabet~$\Sigma$,
by Theorem~\ref{the:sw} we can obtain an equivalent
$2n^k$-state $1\hntrans$~$T'$. The difference between a $1\hntrans$ and an $\nfa$ is
that the former reads the endmarker, whereas the latter does not read an endmarker.
Hence, $T'$ can be converted to an equivalent $\nfa$~$T''$ as follows.
On input $\Sigma$, the transition functions of $T'$ and $T''$ are defined identically.
The accepting states of~$T''$ are defined to be those states $(s_1,s_2,\ldots,s_k)$
such that there is a component in the $k$-tuple
$\delta'((s_1,s_2,\ldots,s_k),\rightend)$ that is an accepting state in~$T$.
\end{proof}

We obtain a lower bound for the state blow-up in Theorem~\ref{the:trans} by establishing
a lower bound for the state cost of sweep reduction proved in Theorem~\ref{the:sw}.
To this aim, for $n,k>0$, let $L_{n,k}$ be the unary language
$$
L_{n,k}=\sett{a^{c\cdot n^k}}{c\geq 0}.
$$
In \cite{kutrib:2022:dciufst}, an $n$-state $k\htrans$ for $L_{n,k}$ is provided, whereas
any equivalent $\dfa$ or $\nfa$ needs at least $n^k$ states. By using
$L_{n,k}$ as a witness language, we can show
the following theorem.

\begin{theorem}\label{the:nfa}
Let $n,k,i>0$ be integers such that $i$ divides $k$.
There exists an $n$-state $k\hntrans$ such that
any equivalent $\frac{k}{i}\hntrans$ cannot have less than $2^{\frac{-i}{k}}n^i$ states.
\end{theorem}

\begin{proof}
Let $T$ be an $n$-state $k\hntrans$.
Suppose by way of contradiction that we could always design an equivalent
$\ntrans$ $T'$ with
$\frac{k}{i}$ sweeps and $s<2^{\frac{-i}{k}}n^i$ states.
By using the result of Theorem~\ref{the:trans}, we can obtain from $T'$ an equivalent
$\nfa$ with $2s^\frac{k}{i}<2\cdot \left(2^{\frac{-i}{k}}n^i\right)^{\frac{k}{i}}= n^k$
states.
By applying this approach in particular to the $n$-state $k\htrans$ recalled above for the
language $L_{n,k}$,
we could obtain an equivalent $\nfa$ having less than $n^{k}$ states
which is a contradiction.
\end{proof}

The lower bound provided by Theorem~\ref{the:nfa} is general in the sense that
the number of $k$ sweeps can be reduced to any $\frac{k}{i}$ as long as $i$
divides $k$. However, for the special case $i=k$, that is, reducing the sweeps
to one only, we have a better lower bound as stated in the following corollary.

\begin{corollary}
For any integers $n,k>0$, there exists an $n$-state $k\hntrans$ which cannot be
converted to an equivalent $\nfa$ with less than $n^k$ states.
\end{corollary}

\begin{proof}
The language $L_{n,k}$ is accepted by some $n$-state $k\htrans$
but any $\nfa$ for $L_{n,k}$ has at least $n^k$ states.
\end{proof}\vspace*{-2.6mm}

 We conclude this section by discussing the state blow-up of turning constant sweep
bounded $\ntrans$s into $\dfa$s, i.e., the cost of removing both nondeterminism and sweeps
at once.

\begin{theorem}\label{the:dfa}
Let $n,k>0$ be integers. Every $n$-state $k\hntrans$ can be converted to an
equivalent $\dfa$ with at most $2^{2n^k}$ states.
\end{theorem}

\begin{proof}
The result follows by first converting, according to Theorem~\ref{the:trans},
the $n$-state $k\hntrans$ into an equivalent $2n^k$-state
$\nfa$ which, in turn, is converted to an equivalent $2^{2n^k}$-state $\dfa$ by the
usual powerset construction (see, e.g.,~\cite{Hopcroft:1979:itatlc:book}).
\end{proof}

A lower bound for the state blow-up in Theorem~\ref{the:dfa} can be proved by
considering the following language for any $n,k>1$:
$$
E_{n,k}=\sett{ubv}{u,v\in\set{a,b}^*,\, |v|=c\cdot n^k-1 \text{ for }c>0}.
$$

\begin{theorem}
For any integers $n>1$ and $k>0$, there is an $(n+1)$-state $k\hntrans$ which cannot be
converted to an equivalent $\dfa$ with less than $2^{n^k}$ states.
\end{theorem}

\begin{proof}
First we construct an $(n+1)$-state $k\hntrans$
$T=\langle Q,\{a,b\},\Delta,q_0,\rightend_0,\delta,F\rangle$
that accepts the language~$E_{n,k}$. We set
$Q=\{q_0,q_1,q_2,\ldots,q_n\}$,
\mbox{$\Delta=\{a,b,\blank, 1,2,\ldots,n,{\bar
    n},\rightend_0,\rightend_1,\dots,\rightend_{k-1}\}$,}
and
\mbox{$F=\{q_n\}$.}

\medskip
Basically, $T$ processes an input string as follows.
In a first sweep, $T$ reads and blanks the input in its initial state $q_0$, where it guesses on every
input symbol $b$ whether it separates the prefix $u$ from the suffix $v$. If
it does not find a $b$ or never guesses in the affirmative, $T$ reaches the
right endmarker in state $q_0$ and halts rejecting. Otherwise, beginning with
the $b$ transducer $T$ starts to check whether the remaining input length
is a multiple of $n$ by rewriting it as a sequence of consecutive
blocks of the form
$12\cdots n$, followed by $\rightend_1$.
If and only if $T$ reaches the endmarker in state $q_1$ the test was
positive and the computation continues.

\medskip
$
\begin{array}[t]{r@{\ }rcl}
(1) & \delta(q_0,x) &=& \{(q_0, \blank)\} \text{ if } x\in\{a,\blank\}\\
(2) & \delta(q_0,b) &=& \{(q_0, \blank), (q_2, 1)\}\\[1mm]
(3) & \delta(q_i,x) &=& \{(q_{i+1}, i)\} \text{ if } x\in\{a,b,n\}, \mbox{for $1\le i\le n-1$}\\
(4) & \delta(q_n,x) &=& \{(q_1, n)\} \text{ if } x\in\{a,b,n\}\\[1mm]
(5) & \delta(q_1,\rightend_i) &=& \{(q_0, \rightend_{i+1}\}, \mbox{for $0\le i\le k-2$}\\
\end{array}
$
\medskip

\noindent
In the second sweep, after reaching the first $1$ in state $q_0$, transducer $T$
checks whether the length of the remaining input string
is a multiple of $n^2$ by rewriting $n$ consecutive blocks $12\cdots n$ with
the block $1^n\ 2^n\ \cdots\ (n-1)^n\ {\bar n}^{n-1}n$,
followed by $\rightend_2$.
If and only if $T$ reaches the endmarker in state $q_1$ the test was
positive and the computation continues.

\medskip
$
\begin{array}[t]{r@{\ }rcl}
(6) & \delta(q_0,1) &=& \{(q_1, 1)\}\\[1mm]
(7) & \delta(q_i,x) &=& \{(q_{i}, i)\} \text{ if } x\in\{1,2,\ldots,n-1,\bar n\}, \mbox{for $1\le i\le n-1$}\\
(8) & \delta(q_n,x) &=& \{(q_n, \bar{n})\} \text{ if } x\in\{1,2,\ldots,n-1,\bar n\}\\
\end{array}
$

\medskip
\noindent
In the third sweep, after reaching the first $1$ in state $q_0$, transducer
$T$ checks whether the length of the remaining input string
is a multiple of~$n^3$ by rewriting $n$ consecutive blocks \mbox{$1^n\ 2^n\ \cdots\ (n-1)^n\ {\bar n}^{n-1}n$} with
the block $1^{n^2}\ 2^{n^2}\ \cdots\ (n-1)^{n^2}\ {\bar n}^{{n^2}-1}n$,   followed by $\rightend_3$.
If and only if $T$ reaches the endmarker in state $q_1$ the test was
positive and the computation continues.

\medskip
This behavior is iterated and, according to {Transition}~(9), the input string
is easily seen to be accepted after the $k$th sweep only, upon reading $\rightend_{k-1}$.

\medskip
$
\begin{array}[t]{r@{\ }rcl}
(9) & \delta(q_1,\rightend_{k-1}) &=& \{(q_n, \rightend_{k-1})\}\\
\end{array}
$

\medskip
It remains to be shown that language~$E_{n,k}$ cannot be accepted by a
$\dfa$ with less than $2^{n^k}$ states.
Suppose by contradiction that there is a $\dfa$ $A$ which accepts $E_{n,k}$ with less than
$2^{n^k}$ states.
Since there are $2^{n^k}$ different words of length $n^k$ over alphabet
$\{a,b\}$ and the $\dfa$ $A$ has less than~$2^{n^k}$ states, at least two
of these words drive $A$ into the same state.
Say the words are \mbox{$w_1=s_1s_2\cdots s_{n^k}$} and $w_2=t_1t_2\cdots t_{n^k}$.
Since $w_1\neq w_2$, there is some $1\leq i\leq n^k$ such that $s_i\neq t_i$.
Without loss of generality, we may assume $s_i=b$ and $t_i=a$.

Since $A$ is in the same state after processing $w_1$ and $w_2$, it is in the
same state after processing $w_1a^{i-1}$ and $w_2a^{i-1}$. We have
$n^k\leq |w_1a^{i-1}|=|w_2a^{i-1}|=n^k+(i-1)\leq 2n^k-1$.
Moreover, $w_1a^{i-1}=s_1s_2\cdots s_{i-1} b s_{i+1}\cdots s_{n^k} a^{i-1}$
and, thus, $w_1a^{i-1}$ belongs to $E_{n,k}$. On the other hand,
$w_2a^{i-1}=t_1t_2\cdots t_{i-1} a t_{i+1}\cdots t_{n^k} a^{i-1}$
and, thus, $w_wa^{i-1}$ does not belong to $E_{n,k}$. This is a contradiction
to the assumption that $A$ has less than $2^{n^k}$ states.
\end{proof}

\subsection{Decidability questions}

Since every $k\hntrans$ can be converted into an equivalent $\dfa$ by Theorem~\ref{the:dfa}
it is clear that all decidable questions for $\dfa$s are also decidable for $k\hntrans$s.
It has been shown in~\cite{kutrib:2022:dciufst} that for $k\htrans$s the questions of testing emptiness,
universality, finiteness, infiniteness, inclusion, or equivalence are
\textsf{NL}-complete.
Here, we will investigate the computational complexity of these questions for
$k\hntrans$s and will obtain that the questions of emptiness,
finiteness, and infiniteness are \textsf{NL}-complete,
whereas the questions of universality,
inclusion, and equivalence turn out to be \textsf{PSPACE}-complete.
So, for $k\hntrans$s the questions of testing emptiness,
universality, finiteness, infiniteness, inclusion, or equivalence
have the same computational complexity as for
$\nfa$s (see, e.g., \cite{Jones:1975:sbracp,garey:1979:npbook}
and~\cite{holzer:2011:daccofaas} for a survey).

\begin{theorem}\label{thm:comp:complexity:1}
Let $k \ge 1$ be an integer. Then
for $k\hntrans$s the problems of testing emptiness,
finiteness, and infiniteness are {\em \textsf{NL}}-complete.
\end{theorem}

\begin{proof}
First, we show that the problem of non-emptiness belongs to~\textsf{NL}.
Since \textsf{NL} is closed under complementation, emptiness belongs to
\textsf{NL} as well. We describe a two-way nondeterministic Turing machine~$M$
which receives an encoding $\cod(A)$ of some $k\hntrans$ $A$ on its
read-only input tape and accepts the input if and only if $A$ accepts a non-empty
language while the space used on its working tape is bounded by
$O(\lg{} |\cod(A)|)$. Then, the work space is bounded by $O(\lg n)$,
with $n$ being the maximum among the number of states in~$A$,
the size of the input alphabet of $A$, and the size of the output alphabet of $A$,
since all parameters are part of the encoding of $A$ on the input tape of $M$.

It is shown in Theorem~\ref{the:trans} which is based on Theorem~\ref{the:sw}
that~$A$ can be converted into an
equivalent $\nfa$ with at most $2n^k$ states. This construction basically
simulates $k$ sweeps in one sweep and keeps in parallel track of the $k$ states
occuring in each sweep.

The basic idea for the Turing machine~$M$ is to guess a word and to check
``on the fly'' whether it is accepted by $A$.
To simulate the $k$ sweeps of~$A$ on the guessed input we apply the construction
described above.
Thus, we have to keep track of the~$k$ current states which are obtained by basically applying
the nondeterministic procedure described in the proof of Theorem~\ref{the:sw}.
A close look on that procedure shows that it can be implemented using no more
than $O(\lg n)$ tape cells.
Moreover, it is clear that each state along a sweep can be represented by $O(\lg n)$ tape cells.
Hence, the current states of the~$k$ sweeps of~$A$ are altogether represented by
$O(\lg n)$ tape cells.

\medskip
Now, the Turing machine $M$ guesses one input symbol~$a$ and updates all stored states
of~$A$. This behavior is iterated until either~$A$ halts or $A$ halts after having
guessed the endmarker. Then, $M$ halts accepting if~$A$ accepts and halts rejecting
in any other case. Thus, $M$ decides the non-emptiness of~$A$
using an amount of tape cells which is at most logarithmic in the
length of the input.

To show that the problem of testing infiniteness belongs to~\textsf{NL}
we basically use the construction for non-emptiness. Additionally,
we count the length of the guessed word and accept only if the underlying
$k\hntrans$ $A$ would have accepted an input of length at least $2n^k$.
If $L(A)$ is infinite, then~$A$ clearly accepts an input of length at least $2n^k$.
On the other hand, if $A$ accepts an input of length at least $2n^k$,
then the equivalent $\nfa$ with at most $2n^k$ states according to Theorem~\ref{the:trans}
accepts an input of length at least $2n^k$ as well. But this means that there is an
accepting computation in the $\nfa$ in which at least one
state is entered twice which implies that an infinite language is accepted.
Thus, it remains to be argued that the counting up to $2n^k$ can be realized in
logarithmic space. However, we can implement on $M$'s working tape a
binary counter that counts up to~$2n^k$. With the usual construction this needs at most
\mbox{$O(\lg 2n^k)=O(\lg n)$} tape cells.
Altogether, we obtain that the problem of testing infiniteness belongs to~\textsf{NL}.
Since \textsf{NL} is closed under complementation the problem of testing finiteness
belongs to~\textsf{NL} as well.

\medskip
The hardness results follow directly from the hardness results for
$\nfa$s, which are basically shown in~\cite{Jones:1975:sbracp}, considering the
fact that every $\nfa$ $N$ with $n$ states can be converted into an
equivalent $k\hntrans$ $N'$ simulating~$N$ and having the same $n$~states
plus an additional accepting state which is entered exactly when $N'$ reads
the endmarker at the end of the first sweep and the previous simulation of~$N$
ended up in an accepting state of $N$. Obviously, the latter construction
can be realized in deterministic logarithmic space.
\end{proof}

\begin{theorem}\label{thm:comp:complexity:2}
Let $k \ge 1$ be an integer. Then
for $k\hntrans$s the problems of testing universality,
inclusion, and equivalence are {\em \textsf{PSPACE}}-complete.
\end{theorem}

\begin{proof}
To show that the problems are in \textsf{PSPACE} it suffices to show that
the conversion of a $k\hntrans$ into an equivalent $\nfa$ is possible
in $\textsf{P} \subseteq \textsf{PSPACE}$, since the corresponding problems
for $\nfa$s are known to be in \textsf{PSPACE}.
We consider a deterministic Turing machine which receives
an encoding $\cod(A)$ of a $k\hntrans$s $A$ and
denote with $n$ the maximum among the number of states in~$A$,
the size of the input alphabet of $A$, and the size of the output
alphabet of $A$.
Hence, we have to show that the time is bounded by some polynomial in $n$.

To estimate the time complexity of the conversion procedure described in
Theorem~\ref{the:sw} and Theorem~\ref{the:trans} we note that the most costly
action is that a nondeterministic
procedure which determines the successor states and the output symbol
is performed for each combination from $Q' \times (\Sigma \cup \Delta')$.
Since each run of the nondeterministic procedure can be realized with an amount of space
that is logarithmic in~$n$, each run can be realized by a deterministic procedure
with an amount of time that is polynomial in~$n$.
Since $|Q'\times (\Sigma \cup \Delta')| \in O(n^k)$, we obtain that the time complexity
of the conversion is in \textsf{P}.

The hardness results follow again directly from the hardness results for
$\nfa$s and the fact that
every $\nfa$ $N$ with $n$ states can be converted in deterministic polynomial
time into an
equivalent $k\hntrans$ simulating~$N$ and having $n+1$~states.
\end{proof}

\section{An infinite sweep hierarchy}\label{sec:infinite-hierarchy}

We now consider $s(n)\hntrans$s where $s(n)$ is a non-constant
function. In~\cite{kutrib:2022:dciufst} it is proved that $o(\lg n)$ sweep bounded
{\trans}s accept regular languages only, and that such a logarithmic
sweep lower bound is tight for nonregular acceptance.
Then, a three-level proper language hierarchy is established, where
$O(n)$ sweeps are better than~$O(\sqrt{n})$ sweeps which, in turn,
are better than $O(\lg n)$ sweeps for $\trans$.
Here, we extend the hierarchy to infinitely many levels for both
$\trans$s and~$\ntrans$s.

Let $f:\mathbb{N} \to\mathbb{N}$ be a non-decreasing function. Its inverse is defined as the function
$
f^{-1}(n)= \min\{\,m\in \mathbb{N}\mid f(m)\geq n\,\}.
$
To show an infinite hierarchy dependent on some resources,
where the limits of the resources are given by some functions in
 the length of the input, it is often necessary to control the lengths of
the input so that they depend on the limiting functions. Usually, this is done
by requiring that the functions are \emph{constructible} in a desired sense.
The following notion of construc\-tibility expresses the idea that the
length of a word relative to the length of a prefix is determined by
a function.

\begin{definition}\label{def:constructible}
A non-decreasing computable function
$f\colon\mathbb{N}\to\mathbb{N}$ with $f(n)\geq n$ is said to be \emph{constructible}
if there exists an $s(n)\htrans$ $T$ with $s(n)\in O(f^{-1}(n))$ and
an input alphabet $\Sigma\cupdot\{a\}$, such that
$$
L(T)\subseteq\{\, a^m v\mid m\geq 1, v\in\Sigma^*, \, |v|=f(m)\,\}
$$
and such that, for all $m\geq 1$, there exists a word of the form $a^mv$ in $L(T)$.
The $s(n)\htrans$ $T$ is said to be a \emph{constructor} for $f$.
\end{definition}

In order to show that the class of functions that are constructible in this
sense is sufficiently rich to witness an infinite dense hierarchy, we next
show that it is closed under addition and multiplication.

\begin{proposition}\label{prop:add-constructible}
Let $f\colon\mathbb{N}\to\mathbb{N}$ and $g\colon\mathbb{N}\to\mathbb{N}$
be two constructible functions.
Then the function $f+g$ is constructible as well.
\end{proposition}

\begin{proof}
Let $T_f$ be a constructor for $f$ and $T_g$ be a constructor for $g$.
We may safely  assume that the input alphabet of $T_f$ is
$\Sigma_f\cupdot\{a\}$ and the input alphabet of~$T_g$ is
$\Sigma_g\cupdot\{a\}$, where $\Sigma_f$ and $\Sigma_g$ are disjoint.
Now, the idea for designing a constructor $T$ for $f+g$
is to have
$$
L(T)\subseteq \{\, a^m v_fv_g\mid m\geq 1,\, v_f\in\Sigma_f^*,\,
|v_f|=f(m),\, v_g\in\Sigma_g^*,\, |v_g|=g(m)\,\}.
$$
The constructor $T$ splits its input into two tracks.
In any sweep, on the first track the constructor $T_f$ is simulated, whereby
the symbols from $\Sigma_g$, that is, the factor $v_g$ is ignored.
Similarly, on the second track the constructor $T_g$ is simulated.
If one of the simulations halts, only the other simulation is continued.
Now, $T$ accepts if and only if both simulations end accepting.

Clearly, the language $L(T)$ is
a subset
as desired. In addition,
for any $m\geq 1$, there exists a word of the form $a^mv_fv_g$ in $L(T)$.
On such a word, the constructor~$T$ has a sweep complexity which is in
$O(\max\{f^{-1}(m+f(m)),g^{-1}(m+g(m))\})$. Since $f(n)\geq n$ and
$g(n)\geq n$ we conclude that the sweep complexity is in $O(m)$. So,
the sweep complexity $s(n)$ of $T$ is $s(n)\in O((f+g)^{-1}(n))$,
which proves the proposition.
\end{proof}

\begin{proposition}\label{prop:mult-constructible}
Let $f\colon\mathbb{N}\to\mathbb{N}$ and $g\colon\mathbb{N}\to\mathbb{N}$
be two constructible functions.
Then the function $f\cdot g$ is constructible as well.
\end{proposition}

\begin{proof}
Let $T_f$ be a constructor for $f$ and $T_g$ be a constructor for $g$.
As in the proof of Proposition~\ref{prop:add-constructible},
we may safely assume that the input alphabet of $T_f$ is
$\Sigma_f\cupdot\{a\}$ and the input alphabet of~$T_g$ is
$\Sigma_g\cupdot\{a\}$, where $\Sigma_f$ and $\Sigma_g$ are disjoint.
Now, the idea for designing a constructor $T$ for $f\cdot g$
is to have
\[
L(T) =\{\, a^m x_1 v_1 x_2 v_2 \cdots x_{g(m)}v_{g(m)}\mid
 m\geq 1,\,
a^mv_i\in L(T_f),\, a^mx_1x_2\cdots x_{g(m)}\in L(T_g)\,\}.
\]

To this end, $T$ splits its input into two tracks.
In any sweep, on the first track the constructor~$T_f$ is simulated, whereby
the symbols from $\Sigma_g$ are ignored.
Moreover, the computation of $T_f$ is independently simulated on any maximal
factor $v_i$
of symbols from $\Sigma_f$. These simulations are all started in the
state in which $T_f$ enters the first symbol after the prefix $a^m$.
On the second track the constructor~$T_g$ is simulated
on symbols of $\Sigma_g$, whereby
the symbols from $\Sigma_f$ are ignored.
If a simulation halts, only the other simulations are continued.
Now, $T$ accepts if and only if all simulations end accepting.

Since $T$ can verify the correct form of the input in an initial sweep,
the language $L(T)$ is
a subset
as desired. In addition,
for any $m\geq 1$, there exists a word
$a^m x_1 v_1 x_2 v_2 \cdots x_{g(m)}v_{g(m)}$ in $L(T)$.
As before, on such a word, $T$ has sweep complexity
$O(\max\{f^{-1}(m+f(m)),g^{-1}(m+g(m))\})$. Since $f(n)\geq n$ and
$g(n)\geq n$ we conclude that the sweep complexity is in $O(m)$. So,
the sweep complexity $s(n)$ of $T$ is $s(n)\in O((f\cdot g)^{-1}(n))$.
\end{proof}

In~\cite{kutrib:2022:dciufst} it is shown that the unary language
$L_\subtext{uexpo}=\{\, a^{2^k}\mid k\geq 0\,\}$ is accepted by
some $s(n)\htrans$ with $s(n)\in O(\lg n)$. The construction can
straightforwardly be extended to show that the function $f(n)=2^n$
is constructible.

Moreover, again from~\cite{kutrib:2022:dciufst}, we know that
\mbox{$L_\subtext{eq}=\{\,u\dollar v\mid u\in \Sigma_1^*, v\in\Sigma_2^*,
  \text{ and } |u|=|v|\,\}$} is a language
accepted by some $s(n)\htrans$ with $s(n)\in O(n)$,
where $\Sigma_1$ is an alphabet not containing the symbol $\dollar$
and $\Sigma_2$ is an arbitrary alphabet. Even in this case, only a tiny modification
shows that the identity function is constructible.
These facts together with Proposition~\ref{prop:mult-constructible}
yield, in particular, that the function $f(n)=n^x$
is constructible for all positive integers~$x$.

In what follows, we will use the fact, proved in \cite{kutrib:2022:dciufst}, that the copy language with center marker
\mbox{$\{\, u\dollar u\mid u\in\{a,b\}^*\,\}$} is accepted by
some $s(n)\htrans$ satisfying $s(n)\in O(n)$.
The next theorem provides some language
that separates the levels of the hierarchy.

\begin{theorem}\label{theo:positive-hier}
Let $f\colon\mathbb{N}\to\mathbb{N}$ be a constructible function,
$T_f$ be a constructor for~$f$ with input alphabet $\Sigma\cupdot\{a\}$,
and $b$ be a new symbol not belonging to $\Sigma\cupdot\{a\}$.
Then language
$$
L_f=\{\, u\dollar u v \mid u\in \{a,b\}^*,\, v\in\Sigma^*,\,
a^{2|u|+1} v \in L(T_f)\, \}
$$
is accepted by some $s(n)\htrans$ with $s(n)\in O(f^{-1}(n))$.
\end{theorem}

\begin{proof}
Since the suffix $v$ of a word $w\in L_f$ must be the suffix of
some word $a^{2|u|+1} v$ in $L(T_f)$, we have that
$|v|=f(2|u|+1)$ and \mbox{$|w|=2|u|+1+f(2|u|+1)$.} Since
$s(n)$ is claimed to be of order $O(f^{-1}(n))$, an $s(n)\htrans$
accepting $L_f$ may perform at least $O(2|u|+1)$ many sweeps.

An $s(n)\htrans$ $T$ accepting $L_f$ essentially combines
in parallel the acceptors for the copy language with center marker
and the language $L(T_f)$.
To this end, $T$ establishes two tracks in its output.
On the first track $T$ simulates an acceptor
for the copy language $\{\, u\dollar u\mid u\in\{a,b\}^*\,\}$,
where the first symbol of $\Sigma$ (i.e., the first symbol of $v$)
acts as endmarker. In this way, the prefix~$u\dollar u$ is verified.
The result of the computation is written to the output track.
This task takes $O(2|u|+1)$ sweeps.
On the second track $T$ simulates the constructor $T_f$,
where all symbols up to the first symbol of $\Sigma$
(i.e., all symbols of the prefix $u\dollar u$) are treated
as input symbols~$a$.
In this way, $T$ verifies that $|v|=f(2|u|+1)$.
The result of the computation
is written to the output track.
This task takes $O(2|u|+1)$ sweeps.

Finally, $T$ rejects whenever one of the above simulations ends rejecting.
Instead, $T$ accepts if it detects positive simulation results of the two
tasks on the tracks.
\end{proof}

To show that the witness language
$L_f$ of Theorem~\ref{theo:positive-hier} is not accepted by any
$s(n)\hntrans$ with $s(n)\in o(f^{-1}(n))$,
we use Kolmogorov complexity and incompressibility arguments.
General information on this technique
can be found, for example, in the textbook~\cite[Ch.~7]{li:1993:itkca:book}.
Let $w\in \{a, b\}^*$ be an arbitrary binary string. Its Kolmogorov
complexity~$C(w)$ is defined to be the minimal size of a
binary program (Turing machine) describing~$w$. The following key fact
for using the incompressibility
method is well known:
there exist binary strings~$w$ of \emph{any} length such that $|w| \leq C(w)$.

\begin{theorem}\label{theo:negative-hier}
Let $f\colon\mathbb{N}\to\mathbb{N}$ be a constructible function,
$T_f$ be a constructor for~$f$ with input alphabet $\Sigma\cupdot\{a\}$,
and $b$ be a new symbol not belonging to $\Sigma\cupdot\{a\}$.
Then language
$$
L_f=\{\, u\dollar u v \mid u\in \{a,b\}^*,\, v\in\Sigma^*,\,
a^{2|u|+1} v \in L(T_f)\, \}
$$
cannot be accepted by any $s(n)\hntrans$ with $s(n)\in o(f^{-1}(n))$.
\end{theorem}

\begin{sloppypar}
\begin{proof}
Contrarily, let us assume that  some $s(n)\hntrans$
\mbox{$T=\langle Q,\Sigma\cupdot\{a,b\},\Delta,q_0,\rightend,\delta,F \rangle$}
with sweep complexity $s(n) \in o(f^{-1}(n))$ accepts~$L_f$.

We choose a word $u \in \{a,b\}^*$ long enough such that
$C(u)\geq |u|$. Then, we consider an accepting computation of $T$ on
$u\dollar u v$, and derive a contradiction by
showing that $u$ can be compressed via $T$.
To this end, we describe a program~$P$ which reconstructs $u$
from a description of $T$, the length $|u|$, and the sequence
of the $o(f^{-1}(n))$ many states $q_1,q_2,\dots, q_r$ entered along
the accepting computation at that moments when
$T$ reads the first symbol after the $\dollar$ along its $o(f^{-1}(n))$
sweeps, $n$ being the total length of the input.

\medskip
Basically, the program~$P$ takes the length $|u|$ and
enumerates the finitely many words $u'v'$ with
$u'\in \{a,b\}^{|u|}$, $v'\in\Sigma^*$, and $|v'|=f(2|u|+1)$.
Then, for each word in the list,
it simulates by dovetailing all possible
computations of $T$ on~$u'v'$ where, in particular,
it simulates the $o(f^{-1}(n))$ successive partial sweeps of $T$
on~$u'v'$, where the $i$th sweep is started in state $q_i$
for $1\leq i\leq r$. If the simulation ends accepting,
we know that $u\dollar u'v'$ belongs to $L_f$ and, thus,
$u'=u$.

Let us consider the Kolmogorov complexity of~$u$. Let $|T|$ denote the constant size
of the description of~$T$, and~$|P|$ denote the constant size of the
program~$P$ itself. The binary description of the length $|u|$ takes
$O(\lg(|u|))$ bits. Each state of $T$ can be encoded by
$O(\lg(|Q|))$ bits.
So, we have
$$
C(u) \in |P|+|T|+O(\lg (|u|)+o(f^{-1}(n))\cdot O(\lg(|Q|))
= O(\lg (|u|))+o(f^{-1}(n)).
$$
Since $n=2|u|+1+f(2|u|+1)$ and $f(n)\geq n$, for all $n\geq 1$,
we have \mbox{$n\in\Theta(f(2|u|+1))$.} So,
we can conclude that $C(u)\in O(\lg (|u|))+o(|u|)=o(|u|)$.
This contradicts our initial assumption $C(u)\!\ge\!|u|$, for $u$ long enough. Therefore,
$T$ cannot accept $L_f$ with sweep complexity $o(f^{-1}(n))$.
\end{proof}
\end{sloppypar}

We would like to remark that, due to our observation
that all functions $f(n)=n^x$ are constructible for $x\geq 1$,
it is an easy application of the above theorems to obtain
the following infinite hierarchies with regard to the number of sweeps
both in the deterministic and the nondeterministic case.
Namely:
{\em For every $x \ge 1$ we have that the set of all languages that are accepted by
$s(n)\htrans$s ($s(n)\hntrans$s)
with $s(n) \in O(n^{1/(x+1)})$ is properly included in
the set of all languages that are accepted by $s(n)\htrans$s ($s(n)\hntrans$s)
with $s(n) \in O(n^{1/x})$.}

Finally, we turn to an upper bound of the computational capacity of
$\ntrans$s, where we do not limit the number of sweeps at all.
Yet, we quickly observe that an $\ntrans$  is clearly a restricted version of a linear bounded
automaton (see, e.g., \cite{Hopcroft:1979:itatlc:book}). So, any language accepted by an
$\ntrans$ is context-sensitive. However, we can also show the converse though
the transducers are one-way devices only that, at a glance, cannot transmit
information from right to left. The proof uses the simulation of linear bounded
automata ($\lba$).

For a given $\lba$~$M$, we denote its state set by $Q$ where $q_0$ is the
initial state, by $T$ its tape alphabet containing the endmarkers~$\leftend$
and $\rightend$, and by $\Sigma \subset T\setminus\{\leftend,\rightend\}$
its input alphabet. The set of accepting states is $F\subseteq Q$ and
the transition function $\delta$ maps from $Q\times T$ to the subsets
of $Q\times (T\cup\{-1,+1\})$. So, the $\lba$ $M$ either can rewrite
the current tape cell or move its head to the left ($-1$) or to the right
($+1$). We may safely assume that an $\lba$ never moves its head beyond
the endmarkers, starts with its head on the left endmarker, always halts,
and accepts only on the right endmarker by halting in an
accepting state.

\begin{theorem}\label{theo:all-csl}
Let $L$ be a context-sensitive language. Then $L$ is accepted by
an $\ntrans$.
\end{theorem}

\begin{proof}
Let $L$ be accepted by some nondeterministic $\lba$
$M=\langle Q,\Sigma,T,q_0,\leftend,\rightend,\delta,F \rangle$.
Since $M$ is assumed to be halting, its computations are
finite sequences of configurations passed through. We will construct
an $\ntrans$ $T=\langle Q',\Sigma, \Delta,p_0,\rightend,\delta',F'\rangle$
that simulates $M$ such that these configurations are
successively emitted one in each sweep. So, the number of sweeps
taken by $T$ is one more than the number of steps performed by $M$.

We set $Q'=Q\cup \hat{Q} \cup \{p_0,p_s,p_1,p_+\}$, where
$\hat{Q}=\{\, \hat{q}\mid q\in Q\,\}$ is a disjoint copy of $Q$,
and $F=\{p_+\}$. The basic idea is to represent
a configuration of $M$ by two tracks.
The current tape content of~$M$ is written on the second track, while
the only nonblank cell of the first track is the cell currently scanned
by the input head and its content is the current state of~$M$.
However, since $M$ has a left endmarker that can be visited an
arbitrary number of times, but $T$ does not have a left endmarker,
the first symbol on $T$'s tracks represents the left endmarker
in addition.

So, we set $\Delta=
((Q\cup\{\blank\}) \times \{\leftend\})\times ((Q\cup\{\blank\}) \times T)
\cup ((Q\cup\{\blank\}) \times T) \cup \Sigma
\cup\{\rightend\}$.

\medskip
The initial input of $T$ is of the form $\Sigma^*\rightend$. Here we assume
that the input is non-empty. The construction can straightforwardly be
extended to handle an empty input word as well.
During its first sweep, $T$ splits this input into two tracks.
To this end, for $x\in \Sigma \cup \{\rightend\}$ we define:

\smallskip
$
\begin{array}[t]{r@{\ }rcl}
 (1) & \delta'(p_0,x) &=& \{(p_s, ((q_0,\leftend),(\blank, x)))\}\\
 (2) & \delta'(p_s,x) &=& \{(p_s, (\blank, x))\}\\
\end{array}
$
\smallskip

During successor sweeps, $T$ uses its initial state $p_0$ to get close
to the position of the cell with non-empty first track. In each step
in state $p_0$ the $\ntrans$ $T$ guesses whether it reads the symbol to
the left of the non-empty first track and whether $M$ performs a left move
(the states from $\hat{Q}$ are used for this purpose).
Should $T$ read the non-empty first track in state $p_0$ it has guessed that
$M$ does not perform a left move and simulates the stationary or right move
with the help of the states from $Q$. Subsequently, $T$ enters state $p_1$
to reach the right endmarker for a new sweep. So, for
$x\in T\setminus\{\leftend\}$ we define:

\smallskip
$
\begin{array}[t]{r@{\ }rcl}
 (3) & \delta'(p_0,((\blank, \leftend),(\blank,x))) &=&
\{(p_0,((\blank,\leftend),(\blank, x)))\} \cup
\{\, (\hat{q}, ((\blank,\leftend),(q, x))) \mid q\in Q\,\}\\
(4) & \delta'(p_0,(\blank,x)) &=&
   \{(p_0,(\blank, x))\} \cup
  \{\, (\hat{q}, (q, x)) \mid q\in Q\,\}\\
\end{array}
$
\smallskip

Transitions (3) and (4) implement the guessing of a left move of $M$ on the
position coming next, where $M$ enters state $q$.
Once in some state $\hat{q}$ the guess is verified by the following
Transition (5). If the guess was wrong, the transition is undefined and the
computation halts rejecting.

\smallskip
$
\begin{array}[t]{r@{\ }rcl}
(5) & \delta'(\hat{q},(r,x)) &=&
        \{ (p_1, (\blank, x))\} \text{ if } (q, -1)\in \delta(r,x)\\
\end{array}
$
\smallskip

The next transitions implement the simulation of stationary and right moves
of $M$ when $T$ reaches the non-empty first track in state $p_0$.
Let $u\in T\setminus\{\rightend\}$.

\smallskip
$
\begin{array}[t]{r@{\ }rcl}
(6) & \delta'(p_0,(q, u)) &=&
        \{\, (p_1, (r,z)) \mid (r,z)\in \delta(q,u)\,\}\cup
        \{\, (r, (\blank,u)) \mid (r,+1)\in \delta(q,u)\,\}\\

(7) & \delta'(p,(\blank, x)) &=& \{(p_1, (p,x))\}\\
\end{array}
$
\smallskip

Next, $T$ uses state $p_1$ to proceed to the end of the track.

\smallskip
$
\begin{array}[t]{r@{\ }rcl}
(8) & \delta'(p_1,(\blank, x)) &=& \{ (p_1, (\blank,x)\}\\
\end{array}
$
\smallskip

The next transition deals with steps that $M$ may perform on its right
endmarker in state $p_0$. It complements Transition (6).

\smallskip
$
\begin{array}[t]{r@{\ }rcl}
(9) & \delta'(p_0,(q, \rightend)) &=&
        \{\, (p_1, (r,\rightend)) \mid (r,\rightend)\in \delta(q,\rightend)\,\}\cup
        \{\, (p_+, (\blank,\rightend)) \mid q\in F\,\}
\end{array}
$
\smallskip

So, if and only if $M$ halts accepting (which appears only on the right
endmarker) $T$ enters its sole accepting state $q_+$ at the end of a
sweep.
It remains to extend the definition of $\delta'$ for some
cases concerning the first ``double'' symbol untreated so far.

\smallskip
$
\begin{array}[t]{r@{\ }rcl}
 (10) & \delta'(p_0,((q, \leftend),(\blank,x))) &=&
\{\,(p_1,((r,\leftend),(\blank, x))) \mid
(r,\leftend)\in\delta(q,\leftend)\,\}\\
 & & & \mathrel{\cup} \{\,(p_1,((\blank,\leftend),(r, x))) \mid (r,+1)\in\delta(q,x)\,\}\\
 (11) & \delta'(p_0,((\blank, \leftend),(q,x))) &=&
\{\,(p_1,((\blank,\leftend),(r, y))) \mid (r,y)\in\delta(q,x)\,\}\\
 & & & \mathrel{\cup} \{\,(p_1,((r,\leftend),(\blank, x))) \mid (r,-1)\in\delta(q,x)\,\}\\
 & & & \mathrel{\cup} \{\,(r,((\blank,\leftend),(\blank, x))) \mid (r,+1)\in\delta(q,x)\,\}\\
\end{array}
$\\
\strut

\vspace*{-8mm}
\end{proof}

\section{Nondeterminism beats determinism on all levels}\label{sec:ndet-beats-det}

We now turn to compare the computational power of $\trans$s and $\ntrans$s.
Since for sweep bounds of order $o(\lg n)$ both variants accept
regular languages only (see \cite{kutrib:2022:dciufst,Hartmanis:1968:ccottmc}), it remains to consider sweep bounds
beyond $o(\lg n)$. Here, we will show that there exist witness languages
that are accepted by some nondeterministic
$s(n)\hntrans$ with $s(n)\in O(\lg n)$, but cannot be accepted
by any deterministic $s(n)\htrans$ with $s(n)\in o(n)$,
thus separating determinism from nondeterminism for
almost all levels of the sweep hierarchy.

\medskip
For any integer $k\geq 1$, let $\bin_k\colon \set{0,1,2,\dots, 2^k-1}\to \{0,1\}^k$ map
any integer in the range from $0$ to $2^k-1$ to its
binary representation of length $k$, starting from the left with the least significant digit and possibly
completed with zeroes to the right.
E.g., $\bin_4(5)=1010$ and $\bin_4(12)=0011$.
We consider the language
\begin{multline*}
D=\{\,a^kb^{2^k} \bin_k(0)u_0 \bin_k(1)u_1 \cdots \bin_k(2^k-1)u_{2^k-1} \bin_k(i)u_{i} \mid\\
k\geq 2,\, 1\leq i\leq 2^k-1,\, u_j\in\{a,b\}^k \text{ for all } 1\leq j\leq 2^k-1 \,\}.
\end{multline*}

\begin{theorem}\label{theo:det-net-positive}
The language~$D$ can be accepted by an
$s(n)\hntrans$ satisfying {$s(n)\in O(\lg n)$.}
\end{theorem}

\begin{proof}
We sketch the construction of an $s(n)\hntrans$ $T$ that accepts $D$
with $s(n)\in O(\lg n)$.
The basic idea of the construction is to use two output tracks.
So, during its first sweep, $T$ splits the input into two tracks, each one getting
the original input. In addition, $T$ verifies if
the structure of the input is correct, that is, if the input
is of the form $a^+b^+ 0^+\{a,b\}^+(\{0,1\}^+\{a,b\}^+)^+1^+\{a,b\}^+$
with at least two leading $a$'s. If the form is incorrect,
$T$ rejects.

In subsequent sweeps, $T$ behaves as follows. The original input on
the first track is kept but the symbols can be marked, while on the second
track the input is successively
shifted to the right. More precisely, in any sweep the first unmarked
symbol $a$ in the leading $a$-block is marked. In the following
$b$-block, every second unmarked symbol $b$ is marked. In the further course
of the sweep, the leftmost unmarked symbol in any $\{0,1\}$-block as well as in any
$\{a,b\}$-block is marked. On the second track, the input is shifted to the right
by one symbol, whereby the last symbol is deleted and some blank symbol is
added at the left.

Let $k\geq 2$ be the length of the leading $a$-block. When the last of its
symbols is marked, $T$ checks in the further course of the sweep
whether in the following $b$-block exactly one symbol remains unmarked,
and whether in all remaining blocks the last symbol is being marked. Only in
this case the computation continues. In all other cases $T$ halts rejecting.

From the construction so far, we derive that if the computation continues
then all but the second block have the same length, namely, length~$k$.
Moreover, since in the second block every second unmarked symbol has been
marked during a sweep and one symbol is left, the length of the block
is~$2^k$.

Next, $T$ continues to shift the content of the second track to the right
until the $\{0,1\}$-blocks are aligned with their neighboring $\{0,1\}$-blocks
(except for the last one). This takes another~$k$ sweeps.
In the next sweep, $T$ checks if the $\{0,1\}$-block on the second track
represents an integer that is one less than the integer represented by the
aligned block on the first track. This can be done by adding one on
the fly and comparing the result with the content on the first track.
Only if the check is successful, $T$ continues. Otherwise, it halts rejecting.
In the former case, we get that the sequence of $\{0,1\}$-blocks represents
the numbers from $0$ to $2^k-1$ in ascending order.

In the next sweep, $T$ guesses the $\{0,1\}$-block that has to match
the rightmost $\{0,1\}$-block and marks it appropriately. Finally,
this block together with its following $\{a,b\}$-block is symbolwise
compared with the last $\{0,1\}$-block together with its following
$\{a,b\}$-block in another $2k$ sweeps. To this end, we note that~$T$
can detect that the last
block follows when it scans a $\{0,1\}$-block consisting of~$1$'s only.
For the comparison, the symbols can further be marked appropriately.

Now, $T$ accepts only if the guessed $\{0,1\}$-block together
with its following $\{a,b\}$-block match the last $\{0,1\}$-
and $\{a,b\}$-block. Otherwise $T$ rejects.
The construction shows that for any word from~$D$ there is
one accepting computation and that only words from~$D$ are
accepted. So, $T$ accepts $D$.

Altogether, $T$ performs at most
$1+k+k+1+2k\in O(k)$ sweeps. The length of the input is
$k+2^k+(2^k+1)\cdot 2k=O(k2^k)$. Since $\lg(O(k2^k))\in O(k)$,
the $\ntrans$ $T$ obeys the sweep bound $s(n)\in O(\lg n)$.
\end{proof}

To show that the witness language $D$ is not accepted by any
$s(n)\htrans$ with $s(n)\in o(n)$, we use again Kolmogorov complexity
and incompressibility arguments.

\begin{theorem}\label{theo:det-net-negative}
The language~$D$ cannot be accepted by
any $s(n)\htrans$ satisfying {$s(n)\in o(n)$.}
\end{theorem}

\begin{sloppypar}
\begin{proof}
In contrast to the assertion, we assume that some $s(n)\htrans$
\mbox{$T=\langle Q,\{a,b,0,1\},\Delta,q_0,\rightend,\delta,F \rangle$}
with $s(n) \in o(n)$ accepts~$D$.
We choose an integer $k\geq 2$ and a word $u \in \{a,b\}^*$ of length
$k2^k$ satisfying $C(u)\geq |u|$. Now, $u$ is split into
$2^k$ factors $u_0,u_1,\dots, u_{2^k-1}$ of length $k$.
Then, we choose an arbitrary factor $u_i$ and
consider the accepting computation of $T$ on
$$a^kb^{2^k} \bin_k(0)u_0 \bin_k(1)u_1 \cdots \bin_k(2^k-1)u_{2^k-1}
\bin_k(i)u_{i}.$$ We are going to show that $u$
can be compressed via $T$.

\medskip
A program~$P$ reconstructs $u$
from a description of $T$, the number $k$, and the sequence
of the $o(n)$ many states $q_1,q_2,\dots, q_r$ in which
$T$ reads the first symbol of the suffix $\bin_k(i)u_{i}$ as follows.
Since~$T$ is deterministic, this sequence of states
is the same for any suffix $\bin_k(i)u_{i}$, $0\leq i\leq 2^k-1$.
The program~$P$ takes the length $k$ and performs, for all
$\bin_k(i)$, $0\leq i\leq 2^k-1$, the following. It
enumerates the words $v\in\{a,b\}^+$ of length $k$.
For each $v$,
it simulates $o(n)$ successive partial sweeps of $T$
on~$\bin_k(i) v$, where the $j$th sweep is started in state $q_j$,
for $1\leq j\leq r$. If the simulation ends accepting,
we know that~$v$ is the $i$th factor $u_i$ of $u$.
In this way, all factors of $u$ and, thus, $u$ are determined.

\medskip
Let us consider the Kolmogorov complexity of~$u$. Let $|T|$ denote the constant size
of the description of $T$, and $|P|$ denote the constant size of the
program~$P$ itself. The binary description of $k$ takes
$O(\lg(|k|))$ bits. Each state of $T$ can be encoded by
$O(\lg(|Q|))$ bits.
So, we have
$$
C(u) \in |P|+|T|+O(\lg (k))+o(n)\cdot O(\lg(|Q|))
= O(\lg (k))+o(n).
$$
The length $n$ of the input  is
$k+2^k+2|u|+2k$. Since $|u|=k2^k$, we have $n\in \Theta(k2^k)$.
So, we can conclude that  $C(u) \in O(\lg (k))+o(|u|)=o(|u|)$.
This contradicts our initial assumption $C(u)\ge|u|$. So,
$T$ cannot accept $D$ with sweep complexity~$o(n)$.
\end{proof}
\end{sloppypar}


\end{document}